\newif\ifproof

\prooftrue

\documentclass[sigconf]{acmart}
\ifproof
  \renewcommand\footnotetextcopyrightpermission[1]{} 
\fi

\usepackage{booktabs} 

\usepackage[algo2e,vlined]{algorithm2e}

\newcommand{\R}{\mathbb{R}}

\newtheorem{assum}{Assumption}
\newtheorem{problem}[assum]{Problem}
\newtheorem{proposition}[assum]{Proposition}

\graphicspath{{Images/}}

\begin{document}

\copyrightyear{2019}
\acmYear{2019}
\setcopyright{acmlicensed}
\acmConference[HSCC '19]{22nd ACM International Conference on Hybrid Systems: Computation and Control}{April 16--18, 2019}{Montreal, QC, Canada}
\acmBooktitle{22nd ACM International Conference on Hybrid Systems: Computation and Control (HSCC '19), April 16--18, 2019, Montreal, QC, Canada}
\acmPrice{15.00}
\acmDOI{10.1145/3302504.3311808}
\acmISBN{978-1-4503-6282-5/19/04}

\title{TIRA: Toolbox for Interval Reachability Analysis}

\titlenote{Funded in part by the National Science Foundation under Grant CNS-1446145.}

\author{Pierre-Jean Meyer}
\affiliation{%
  \institution{University of California, Berkeley}
  \city{Berkeley}
  \state{CA, USA}
}
\email{pjmeyer@berkeley.edu}

\author{Alex Devonport}
\affiliation{%
  \institution{University of California, Berkeley}
  \city{Berkeley}
  \state{CA, USA}
}
\email{alex_devonport@berkeley.edu}

\author{Murat Arcak}
\affiliation{%
  \institution{University of California, Berkeley}
  \city{Berkeley}
  \state{CA, USA}
}
\email{arcak@berkeley.edu}

\begin{abstract}

This paper presents TIRA, a Matlab library gathering several methods for the computation of interval over-approximations of the reachable sets for both continuous- and discrete-time nonlinear systems.
Unlike other existing tools, the main strength of interval-based reachability analysis is its simplicity and scalability, rather than the accuracy of the over-approximations.
The current implementation of TIRA contains four reachability methods covering wide classes of nonlinear systems, handled with recent results relying on contraction/growth bounds and monotonicity concepts.
TIRA's architecture features a central function working as a hub between the user-defined reachability problem and the library of available reachability methods.
This design choice offers increased extensibility of the library, where users can define their own method in a separate function and add the function call in the hub function.
\end{abstract}

%
%
%


\maketitle

\section{Introduction}
\label{sec intro}
Reachability analysis aims to compute the set of successor states that can be reached by a system given sets of initial states and admissible inputs.
Since an exact computation of a reachable set is not possible for most systems, we rely on methods to over-approximate this set.
Various tools and set representations for these over-approximations have been proposed in the literature, such as zonotopes in CORA~\cite{althoff2015introduction}, support functions in SpaceEx~\cite{frehse2011spaceex}, ellipsoids in the Ellipsoidal Toolbox~\cite{kurzhanskiy2006ellipsoidal}, Taylor models in Flow$^*$~\cite{chen2013flow}, polytopes in Sapo~\cite{dreossi2017sapo} or interval pavings~\cite{jaulin2001applied}.
Other tools such as the Level Set Toolbox~\cite{mitchell2005toolbox} are instead designed to tackle backward reachability problems.

The main common point of the above reachability methods is that their primary focus is to compute a set that over-approximates the actual reachable set as tightly as possible.
While such approaches are particularly interesting to minimize the conservativeness of the over-approximation in simple verification objectives (e.g.\ with safety or reachability specifications), the inherent complexity of the set representations allowing for such tight approximations can make these sets impractical to use if further manipulations are required (e.g.\ saving in memory, intersection with another set).

On the other hand, reachability analysis plays a central role in the field of abstraction-based control synthesis (see e.g.~\cite{moor2002abstraction,coogan2015efficient,reissig2016,meyer2017hierarchical}), where a reachable set over-approximation needs to be computed for each cell of a state space partition and each input value (i.e.\ exponential complexity in the state and input dimensions), and the abstraction is obtained by intersecting these sets with the partition elements.
In addition, existing abstraction tools are limited by their internal reachability algorithms: e.g.\ SCOTS~\cite{scots} relies on the hard-coded growth bound method; PESSOA~\cite{pessoa} cannot handle nonlinear systems unless the user provides their own over-approximation function.
This motivated recent work~\cite{coogan2015efficient,reissig2016,meyer2018sampled,meyer2017hierarchical} on reachability analysis based on the simpler set representation of \emph{multi-dimensional intervals} (also known as \emph{axis-aligned boxes} or \emph{hyper-rectangles}).
While intervals usually result in more conservative over-approximations of the reachable sets, they have useful advantages for the implementation of abstraction-based algorithms: they are fully defined with only two state vectors; their intersection is still an interval; the associated over-approximation methods have very good scalability with a complexity (number of successor computations) at best constant~\cite{reissig2016,moor2002abstraction,meyer2017hierarchical,coogan2015efficient} and at worst linear in the state dimension~\cite{meyer2018sampled}.
Therefore, compared to existing reachability analysis tools, the interval-based methods trade off the accuracy of the over-approximating sets for the simplicity and scalability of the reachability analysis, while still resulting in the tightest possible interval over-approximation for some of these methods~\cite{moor2002abstraction,coogan2015efficient,meyer2018sampled}.

In this paper, we introduce TIRA\,\footnote{\url{https://gitlab.com/pj_meyer/TIRA}}
(Toolbox for Interval Reachability Analysis), a Matlab library gathering several methods to compute interval over-approximations of reachable sets for both continuous- and discrete-time systems.
The primary motivation for the introduction of this tool library is to make publicly available some of the more recent results on interval reachability analysis~\cite{coogan2015efficient,reissig2016,meyer2018sampled,meyer2017hierarchical} and allow external users an easy access to these methods without requiring them to know the theoretical or implementation details.
The architecture of the toolbox features a central function working as a hub between the user-defined reachability problem and the library of available reachability methods.
It takes the initial state and input intervals and returns the over-approximation interval, applying either the method requested by the user, or otherwise picking the most suitable one based on the system properties.
The motivation for this architecture is to offer an easily extensible library, where users can define their own method in a separate function and then add its call in the hub function.

TIRA currently contains four over-approximation methods covering very wide classes of systems: any system with known Jacobian bounds; and any continuous-time system with constant input functions over the time range of the reachability analysis.
The three methods for continuous-time systems are based on contraction/growth bounds~\cite{kapela2009lohner,reissig2016}, mixed-monotonicity~\cite{meyer2017hierarchical}, and sampled-data mixed-monotonicity~\cite{meyer2018sampled}.
The unique method for discrete-time systems is based on mixed-monotonicity~\cite{meyer2018sampled}.

The paper is organized as follows.
Section~\ref{sec problem} introduces notations and formulates the considered reachability problems.
Section~\ref{sec reachability} gives an overview of the implemented over-approximation methods alongside their main limitations and the relevant literature.
The toolbox architecture is summarized in Section~\ref{sec toolbox}.
Finally, Section~\ref{sec examples} demonstrates the use of TIRA on numerical examples.

\section{Problem formulation}
\label{sec problem}
Let $\R$ and $\R^n$ be the sets of real numbers and $n$-dimensional real vectors, respectively.
$1_n$ and $0_n$ are $n$-dimensional vectors filled with ones and zeros, respectively.
$I_n$ is the $n\times n$ identity matrix.
Given $a,b\in\R^n$, $[a,b]\subseteq\R^n$ denotes the $n$-dimensional interval $\{x\in\R^n~|~a\leq x\leq b\}$, using componentwise inequalities.
Given a set $X\subseteq\R^n$, interval $[a,b]\subseteq\R^n$ is said to be a \emph{tight} interval over-approximation of $X$ if $X\subseteq[a,b]$ and for any strictly included interval $[c,d]\varsubsetneq[a,b]$, we have $X\nsubseteq[c,d]$.

We consider both continuous-time and discrete-time systems with time-varying vector field
\begin{gather}
\dot x = f(t,x,p),\label{eq system ct}\\
x^+ = F(t,x,p),\label{eq system dt}
\end{gather}
with time $t\in\R$, state $x\in\R^{n_x}$ and input $p\in\R^{n_p}$.
For the continuous-time system (\ref{eq system ct}), $\Phi(t;t_0,x_0,\mathbf{p})$ denotes the state (assumed to exist and be unique) reached at time $t\geq t_0$ by system (\ref{eq system ct}) starting from initial state $x_0\in\R^{n_x}$ at time $t_0\in\R$ and under piecewise continuous input function $\mathbf{p}:[t_0,+\infty)\rightarrow\R^{n_p}$.
For a constant input function $\mathbf{p}\equiv p\in\R^{n_p}$ over the time range $[t_0,t]$, we write $\Phi(t;t_0,x_0,p)$.
$\Phi$ is evaluated through Runge-Kutta methods and the associated errors are currently neglected in TIRA.

\begin{problem}[Continuous-time reachability]
\label{pb ct}
Given time range $[t_0,t_f]\subseteq\R$, interval of initial states $[\underline{x},\overline{x}]\subseteq\R^{n_x}$ and interval of input values $[\underline{p},\overline{p}]\subseteq\R^{n_p}$, find an interval in $\R^{n_x}$ over-approximating the reachable set of (\ref{eq system ct}) defined as:
\begin{multline*}
R(t_f;t_0,[\underline{x},\overline{x}],[\underline{p},\overline{p}])=\\
\{\Phi(t_f;t_0,x_0,\mathbf{p})~|~x_0\in[\underline{x},\overline{x}],\mathbf{p}:[t_0,t_f]\rightarrow[\underline{p},\overline{p}]\}.
\end{multline*}
\end{problem}

\begin{problem}[Discrete-time reachability]
\label{pb dt}
Given initial time $t_0\in\R$, interval of initial states $[\underline{x},\overline{x}]\subseteq\R^{n_x}$ and interval of input values $[\underline{p},\overline{p}]\subseteq\R^{n_p}$, find an interval in $\R^{n_x}$ over-approximating the reachable set of (\ref{eq system dt}) defined as:
$$R(t_0,[\underline{x},\overline{x}],[\underline{p},\overline{p}])=\{F(t_0,x_0,p)~|~x_0\in[\underline{x},\overline{x}],p\in[\underline{p},\overline{p}]\}.$$
\end{problem}

All over-approximation methods summarized in the next section rely on the Jacobian (assuming a continuously differentiable vector field) and sensitivity matrices of systems (\ref{eq system ct}) and (\ref{eq system dt}).
The state and input Jacobian matrices of (\ref{eq system ct}) are given by the partial derivatives $J_x(t,x,p)=\frac{\partial f(t,x,p)}{\partial x}$ and $J_p(t,x,p)=\frac{\partial f(t,x,p)}{\partial p}$, respectively.
The Jacobian matrices of (\ref{eq system dt}) are similarly obtained by replacing $f$ by $F$.
For continuous-time systems (\ref{eq system ct}) with constant input functions on $[t_0,t_f]$, we further define the sensitivity of the trajectories $\Phi$ to variations of the initial state $S_x(t_f;t_0,x_0,p)=\frac{\partial \Phi(t_f;t_0,x_0,p)}{\partial x_0}$ and to variations of the input value $S_p(t_f;t_0,x_0,p)=\frac{\partial \Phi(t_f;t_0,x_0,p)}{\partial p}$.

\section{Reachability methods}
\label{sec reachability}
In this section, we give an overview of the four methods currently implemented in TIRA for the over-approximation of the reachable set of system (\ref{eq system ct}) or (\ref{eq system dt}) by an interval.
For more in-depth descriptions and proofs, the reader is referred to the papers mentioned in each of the subsections below.

\subsection{Contraction/growth bound}
\label{sub growth bound}
This method holds various names in the literature and can be seen as a particular case of the results in~\cite{kapela2009lohner} based on logarithmic norms, an extension to time-varying systems of the growth bound approach in~\cite{reissig2016}, or an extension to systems with inputs of the componentwise contraction results in~\cite{ArcMai18EALfest}.
Let $x^*=\frac{\underline{x}+\overline{x}}{2}\in\R^{n_x}$ and $[x]=\frac{\overline{x}-\underline{x}}{2}\in\R^{n_x}$ be the center and half-width of the initial state interval $[\underline{x},\overline{x}]$, respectively.
Similarly define $p^*$ and $[p]$ for $[\underline{p},\overline{p}]$.

\paragraph{Requirements and limitations}
The main result of this approach presented below is limited to continuous-time systems (\ref{eq system ct}) with additive input, i.e.\ $n_p=n_x$ and for all $t\in\R$, $x\in\R^{n_x}$, $p\in\R^{n_p}$:
\begin{equation}
\label{eq system ct additive input}
f(t,x,p)=f(t,x,0_{n_p})+p.
\end{equation}
In addition, we assume that we are provided a componentwise contraction/growth matrix defined as follows.
\begin{assum}
\label{assum growth bound}
Given an invariant state space $X\subseteq\R^{n_x}$, there exists $C\in\R^{n_x\times n_x}$ such that for all $t\in[t_0,t_f]$, $x\in X$ and $i,j\in\{1,\dots,n_x\}$ with $j\neq i$  we have:
\begin{equation*}
\begin{cases}
C_{ii}\geq {J_x}_{ii}(t,x,p^*),\\
C_{ij}\geq |{J_x}_{ij}(t,x,p^*)|.
\end{cases}
\end{equation*}
\end{assum}


\paragraph{Method description}
We first define a growth bound function $G:\R_{\geq0}\times\R_{\geq0}^{n_x}\times\R_{\geq0}^{n_x}\rightarrow\R_{\geq0}^{n_x}$: 
\begin{equation}
\label{eq growth bound}
G(\tau,x,p) = e^{C\tau}x+\int_{0}^{\tau}e^{Ct}pdt.
\end{equation}
Then, an interval over-approximation of the reachable set of (\ref{eq system ct additive input}) is obtained by adding and subtracting $G(t_f-t_0,[x],[p])$ to the successor of (\ref{eq system ct additive input}) from the pair $(x^*,p^*)$ of the interval centers.
\begin{proposition}
\label{prop growth bound}
Under Assumption~\ref{assum growth bound} and definition (\ref{eq growth bound}), an over-approximation of the reachable set of (\ref{eq system ct additive input}) in Problem~\ref{pb ct} is given by:
\begin{multline*}
R(t_f;t_0,[\underline{x},\overline{x}],[\underline{p},\overline{p}])\subseteq\\
[\Phi(t_f;t_0,x^*,p^*)-G(t_f-t_0,[x],[p]),\\
\Phi(t_f;t_0,x^*,p^*)+G(t_f-t_0,[x],[p])].
\end{multline*}
\end{proposition}

\paragraph{Remarks}
The following variations of this approach are also available in TIRA.
Firstly, Assumption~\ref{assum growth bound} can be replaced by the existence of a scalar contraction/growth factor $c\in\R$ upper bounding the logarithmic norm (associated to any matrix norm) of $J_x(t,x,p^*)$,
\begin{equation*}
c\geq \lim_{h\rightarrow0^+}\frac{\|I_{n_x}+hJ_x(t,x,p^*)\|-1}{h},\quad\forall t\in[t_0,t_f],~x\in X,
\end{equation*}
which can then be used directly in the growth bound definition (\ref{eq growth bound}) and Proposition~\ref{prop growth bound}, replacing matrix $C$ by scalar $c$~\cite{kapela2009lohner}.

Secondly, for general dynamics (\ref{eq system ct}) without the additive input assumption from (\ref{eq system ct additive input}), Proposition~\ref{prop growth bound} is modified by replacing $[p]$ by a user-provided vector $\tilde p\in\R_{\geq0}^{n_x}$ bounding the influence of the input on the dynamics (using componentwise $\geq$ and $|\cdot|$ operators)~\cite{kapela2009lohner}:
\begin{equation*}
\tilde p\geq |f(t,x,p)-f(t,x,p^*)|,\quad\forall t\in[t_0,t_f],~x\in X.
\end{equation*}

Lastly, for general systems (\ref{eq system ct}), TIRA also allows the user to define their own growth bound function $\tilde G:\R_{\geq0}\times\R_{\geq0}^{n_x}\times\R_{\geq0}^{n_p}\rightarrow\R_{\geq0}^{n_x}$ (replacing $G$ in Proposition~\ref{prop growth bound}) that needs to satisfy (with componentwise inequalities and absolute values)~\cite{reissig2016}:
\begin{equation*}
\begin{cases}
\tilde G(\tau,x,p)\geq\tilde G(\tau,y,q),\quad\forall \tau>0,~x\geq y,~p\geq q,\\
|\Phi(t_f;t_0,x_0,p)-\Phi(t_f;t_0,x^*,p^*)|\leq \tilde G(t_f-t_0,|x_0-x^*|,|p-p^*|),\\
\hfill\forall x_0\in[\underline{x},\overline{x}],~p\in[\underline{p},\overline{p}].
\end{cases}
\end{equation*}

A more general result  allows matrix $C$ to be defined over any partition of the state dimensions $\{1,\dots,n_x\}$ (instead of a partition into $n_x$ elements as in Assumption~\ref{assum growth bound})~\cite{kapela2009lohner}.
This approach is not yet implemented in TIRA but a preliminary algorithm exists in~\cite{ArcMai18EALfest}.

\subsection{Continuous-time mixed-monotonicity}
\label{sub mixed ct}

\paragraph{Requirements and limitations}
Mixed-monotonicity of continuous-time systems (\ref{eq system ct}) is an extension of the monotonicity property~\cite{angeli2003monotone}, where a non-monotone system is decomposed into its increasing and decreasing components~\cite{chu1998mixed}.
A first characterization of a mixed-monotone system relying on the sign-stability of its Jacobian matrices~\cite{coogan2016mixed} was recently relaxed into simply having bounded Jacobian matrices~\cite{yang2018sufficient}, and then used for reachability analysis in~\cite{meyer2017hierarchical}.
The result presented below is a further relaxation of the mixed-monotonicity conditions in~\cite{yang2018sufficient} and~\cite{meyer2017hierarchical}, where the diagonal elements of the state Jacobian are not required to be bounded.
\ifproof
\footnote{The proofs of the new results in this section are provided in Appendix~\ref{appendix proofs}.}
\else
\footnote{The proofs of the new results in this section are provided in the extended version of this paper, available online at \url{}}
\fi

\begin{assum}
\label{assum ct mixed mono}
Given an invariant state space $X\subseteq\R^{n_x}$, there exist $\underline{J_x},\overline{J_x}\in\R^{n_x\times n_x}$ (possibly with $\underline{J_x}_{ii}=-\infty$, $\overline{J_x}_{ii}=+\infty$ for $i\in\{1,\dots,n_x\}$) and $\underline{J_p},\overline{J_p}\in\R^{n_x\times n_p}$ such that for all $t\in[t_0,t_f]$, $x\in X$, $p\in[\underline{p},\overline{p}]$ we have $J_x(t,x,p)\in[\underline{J_x},\overline{J_x}]$ and $J_p(t,x,p)\in[\underline{J_p},\overline{J_p}]$.
\end{assum}

\paragraph{Method description}
Let $J_x^*\in\R^{n_x\times n_x}$ and $J_p^*\in\R^{n_x\times n_p}$ denote the center of $[\underline{J_x},\overline{J_x}]$ and $[\underline{J_p},\overline{J_p}]$, respectively.
We first introduce the decomposition function $g:\R\times\R^{n_x}\times\R^{n_p}\times\R^{n_x}\times\R^{n_p}\rightarrow\R^{n_x}$ defined on each dimension $i\in\{1,\dots,n_x\}$ such that for all $t\in[t_0,t_f]$,  $x,\hat x\in X$ and $p,\hat p\in [\underline{p},\overline{p}]$ we have:
\begin{equation}
\label{eq decomposition function}
g_i(t,x,p,\hat x,\hat p)=f_i(t,\xi^i,\pi^i) + \alpha^i(x-\hat x) + \beta^i(p-\hat p),
\end{equation}
where for each dimension $i\in\{1,\dots,n_x\}$, state $\xi^i=[\xi^i_1;\dots;\xi^i_{n_x}]\in\R^{n_x}$, input $\pi^i=[\pi^i_1;\dots;\pi^i_{n_p}]\in\R^{n_p}$ and row vectors $\alpha^i=[\alpha^i_1,\dots,\alpha^i_{n_x}]\in\R^{n_x}$ and $\beta^i=[\beta^i_1,\dots,\beta^i_{n_p}]\in\R^{n_p}$ are defined according to the Jacobian bounds in Assumption~\ref{assum ct mixed mono} such that for all $j\in\{1,\dots,n_x\}$ and $k\in\{1,\dots,n_p\}$:
\begin{equation}
\label{eq mixed mono variables}
\begin{aligned}
(\xi^i_i,\alpha^i_i)=&\ (x_i,0)\\
(\xi^i_j,\alpha^i_j)=&
\begin{cases}
(x_j,\max(0,-\underline{J_x}_{ij}))&\text{ if }j\neq i\text{ and }{J_x^*}_{ij}\geq0,\\
(\hat x_j,\max(0,\overline{J_x}_{ij}))&\text{ if }j\neq i\text{ and }{J_x^*}_{ij}<0,\\
\end{cases}\\
(\pi^i_k,\beta^i_k)=&
\begin{cases}
(p_k,\max(0,-\underline{J_p}_{ik}))&\text{ if }{J_p^*}_{ik}\geq0,\\
(\hat p_k,\max(0,\overline{J_p}_{ik}))&\text{ if }{J_p^*}_{ik}<0.\\
\end{cases}
\end{aligned}
\end{equation}
Then, consider the dynamical system evolving in $\R^{2n_x}$:
\begin{equation}
\label{eq duplicated system}
\begin{pmatrix}
\dot x\\ \dot{\hat x}
\end{pmatrix}
=h(t,x,p,\hat x,\hat p)=
\begin{pmatrix}
g(t,x,p,\hat x,\hat p)\\g(t,\hat x,\hat p,x,p)
\end{pmatrix},
\end{equation}
whose trajectories from initial state $[x_0;\hat{x_0}]\in\R^{2n_x}$ at time $t_0\in\R$ with constant input $[p;\hat p]\in\R^{2n_p}$ are denoted as $\Phi^h(\cdot;t_0,x_0,p,\hat{x_0},\hat p):[t_0,t_f]\rightarrow\R^{2n_x}$.
Finally, let $\Phi^h_{1\dots n_x}$ and $\Phi^h_{n_x+1\dots 2n_x}$ denote the first and last $n_x$ components of $\Phi^h$, respectively.
Then, an over-approximation of the reachable set of (\ref{eq system ct}) is obtained from the evaluation of a single successor $\Phi^h$ of system (\ref{eq duplicated system}).
\begin{proposition}
\label{prop ct mixed mono}
Under Assumption~\ref{assum ct mixed mono} and definitions (\ref{eq decomposition function}-\ref{eq duplicated system}), an over-approximation of the reachable set of (\ref{eq system ct}) in Problem~\ref{pb ct} is given by:
\begin{multline*}
R(t_f;t_0,[\underline{x},\overline{x}],[\underline{p},\overline{p}])\subseteq\\
[\Phi^h_{1\dots n_x}(t_f;t_0,\underline{x},\underline{p},\overline{x},\overline{p}),
\Phi^h_{n_x+1\dots 2n_x}(t_f;t_0,\underline{x},\underline{p},\overline{x},\overline{p})].
\end{multline*}
\end{proposition}

\paragraph{Remarks}
The mixed-monotonicity definition used in this section encompasses monotonicity~\cite{angeli2003monotone} as a particular case.
Proposition~\ref{prop ct mixed mono} applied to a monotone system thus provides the same result as the reachability method defined for monotone systems in~\cite{moor2002abstraction}.
\begin{proposition}
\label{prop ct mono}
If system (\ref{eq system ct}) is monotone with respect to orthants of $\R^{n_x}$ and $\R^{n_p}$, then Proposition~\ref{prop ct mixed mono} gives the unique tight over-approximating interval of the reachable set of (\ref{eq system ct}) in Problem~\ref{pb ct}.
\end{proposition}

\subsection{Sampled-data mixed-monotonicity}
\label{sub sensitivity}
\paragraph{Requirements and limitations}
This method, presented in~\cite{meyer2018sampled}, corresponds to a discrete-time mixed-monotonicity approach applied to the sampled version of a continuous-time system.
It relies on bounds of the sensitivity matrices and it is an extension of the approach for systems with sign-stable sensitivities in~\cite{xue2017cdc}.
As mentioned in Section~\ref{sec problem}, this approach is limited to systems (\ref{eq system ct}) with constant input functions over the considered time range $[t_0,t_f]$ (sensitivity $S_p$ cannot be defined otherwise).
\begin{assum}
\label{assum sensi}
There exists $\underline{S_x},\overline{S_x}\in\R^{n_x\times n_x}$ and $\underline{S_p},\overline{S_p}\in\R^{n_x\times n_p}$ such that for all initial state $x_0\in [\underline{x},\overline{x}]$ and constant input $p\in[\underline{p},\overline{p}]$ we have $S_x(t_f;t_0,x_0,p)\in[\underline{S_x},\overline{S_x}]$ and $S_p(t_f;t_0,x_0,p)\in[\underline{S_p},\overline{S_p}]$.
\end{assum}

\paragraph{Method description}
Let $S_x^*\in\R^{n_x\times n_x}$ and $S_p^*\in\R^{n_x\times n_p}$ denote the center of $[\underline{S_x},\overline{S_x}]$ and $[\underline{S_p},\overline{S_p}]$, respectively.
For each $i,j\in\{1,\dots,n_x\}$ and $k\in\{1,\dots,n_p\}$, define $\underline{\xi}^i_j,\overline{\xi}^i_j,\alpha^i_j,\underline{\pi}^i_k,\overline{\pi}^i_k,\beta^i_k\in\R$ such that
\begin{equation}
\label{eq sensi variables}
\begin{aligned}
(\underline{\xi}^i_j,\overline{\xi}^i_j,\alpha^i_j)=
\begin{cases}
(\underline{x}_j,\overline{x}_j,\min(0,\underline{S_x}_{ij}))&\text{ if }{S_x^*}_{ij}\geq0,\\
(\overline{x}_j,\underline{x}_j,\max(0,\overline{S_x}_{ij}))&\text{ if }{S_x^*}_{ij}<0,\\
\end{cases}\\
(\underline{\pi}^i_k,\overline{\pi}^i_k,\beta^i_k)=
\begin{cases}
(\underline{p}_k,\overline{p}_k,\min(0,\underline{S_p}_{ik}))&\text{ if }{S_p^*}_{ik}\geq0,\\
(\overline{p}_k,\underline{p}_k,\max(0,\overline{S_p}_{ik}))&\text{ if }{S_p^*}_{ik}<0.
\end{cases}
\end{aligned}
\end{equation}

For all $i\in\{1,\dots,n_x\}$, define the states $\underline{\xi}^i=[\underline{\xi}^i_1;\dots;\underline{\xi}^i_{n_x}]\in\R^{n_x}$, $\overline{\xi}^i=[\overline{\xi}^i_1;\dots;\overline{\xi}^i_{n_x}]\in\R^{n_x}$, inputs $\underline{\pi}^i=[\underline{\pi}^i_1;\dots;\underline{\pi}^i_{n_p}]\in\R^{n_p}$, $\overline{\pi}^i=[\overline{\pi}^i_1;\dots;\overline{\pi}^i_{n_p}]\in\R^{n_p}$ and row vectors $\alpha^i=[\alpha^i_1,\dots,\alpha^i_{n_x}]\in\R^{n_x}$ and $\beta^i=[\beta^i_1,\dots,\beta^i_{n_p}]\in\R^{n_p}$.
Then an over-approximation of the reachable set of (\ref{eq system ct}) is obtained as follows.
\begin{proposition}
\label{prop sensi}
Under Assumption~\ref{assum sensi} and the definitions in (\ref{eq sensi variables}), an over-approximation of the reachable set of (\ref{eq system ct}) in Problem~\ref{pb ct} is given in each dimension $i\in\{1,\dots,{n_x}\}$ by:
\begin{multline*}
R_i(t_f;t_0,[\underline{x},\overline{x}],[\underline{p},\overline{p}])\subseteq\\
[\Phi_i(t_f;t_0,\underline{\xi}^i,\underline{\pi}^i)-\alpha^i(\underline{\xi}^i-\overline{\xi}^i)-\beta^i(\underline{\pi}^i-\overline{\pi}^i),\\
\Phi_i(t_f;t_0,\overline{\xi}^i,\overline{\pi}^i)+\alpha^i(\underline{\xi}^i-\overline{\xi}^i)+\beta^i(\underline{\pi}^i-\overline{\pi}^i)].
\end{multline*}
\end{proposition}

\paragraph{Remarks}
The approach in~\cite{xue2017cdc} restricted to systems with sign-stable sensitivity matrices (i.e.\ $\underline{S_x}_{ij}\geq0$ or $\overline{S_x}_{ij}\leq0$ for all $i,j$) is covered by Proposition~\ref{prop sensi} as the particular case where $\alpha^i=0_{n_x}$ and $\beta^i=0_{n_p}$ for all $i\in\{1,\dots,n_x\}$.
In such case, the interval in Proposition~\ref{prop sensi} is a tight over-approximation of the reachable set.

If the user does not provide sensitivity bounds as in Assumption~\ref{assum sensi}, TIRA offers two methods to compute such bounds (technical details on both methods can be found in~\cite{meyer2018sampled}).
The first one relies on Jacobian bounds similarly to Assumption~\ref{assum ct mixed mono} and applies \emph{interval arithmetic} as in~\cite{althoff2007reachability} to obtain sensitivity bounds guaranteed to satisfy Assumption~\ref{assum sensi}. However, this approach tends to be overly conservative due to being based on global Jacobian bounds.

The second one approximates sensitivity bounds through \emph{sampling and falsification}: first evaluate the sensitivity matrices $S_x$ and $S_p$ for some sample pairs $(x_0,p)\in[\underline{x},\overline{x}]\times[\underline{p},\overline{p}]$; then iteratively falsify the obtained bounds through an optimization problem looking for pairs $(x_0,p)$ whose sensitivities do not belong to the current bounds.
This simulation-based approach does not require any additional assumption and results in much better approximations of the sensitivity bounds, but requires longer computation times and lacks formal guarantees that Assumption~\ref{assum sensi} is satisfied.

\subsection{Discrete-time mixed-monotonicity}
\label{sub mixed dt}

\paragraph{Requirements and limitations}
As highlighted in~\cite{meyer2018sampled}, any discrete-time system (\ref{eq system dt}) can be defined as the sampled version of a continuous-time system (\ref{eq system ct}): $x^+=F(t,x,p)=\Phi(t_f;t,x,p)$ with constant input $p$ over the time range $[t,t_f]$.
Therefore, the approach used in Section~\ref{sub sensitivity} for a sampled continuous-time system can also be applied to a discrete-time system.
The only difference is that conditions on the sensitivity matrices $S_x(t_f)$ and $S_p(t_f)$ of (\ref{eq system ct}) are to be replaced by their equivalent on the Jacobian matrices $J_x$ and $J_p$ of (\ref{eq system dt}).
\begin{assum}
\label{assum dt mixed mono}
There exists $\underline{J_x},\overline{J_x}\in\R^{n_x\times n_x}$ and $\underline{J_p},\overline{J_p}\in\R^{n_x\times n_p}$ such that for all initial state $x_0\in [\underline{x},\overline{x}]$ and input $p\in[\underline{p},\overline{p}]$ we have $J_x(t_0,x_0,p)\in[\underline{J_x},\overline{J_x}]$ and $J_p(t_0,x_0,p)\in[\underline{J_p},\overline{J_p}]$.
\end{assum}

\paragraph{Method description}
Proposition~\ref{prop sensi} is then adapted as follows.
\begin{proposition}
\label{prop dt mixed mono}
Under Assumption~\ref{assum dt mixed mono}, consider $\underline{\xi}^i$, $\overline{\xi}^i$, $\underline{\pi}^i$, $\overline{\pi}^i$, $\alpha^i$, $\beta^i$ defined as in (\ref{eq sensi variables}) but using the Jacobian bounds instead of the sensitivity bounds.
Then, an over-approximation of the reachable set of (\ref{eq system dt}) in Problem~\ref{pb dt} is given in each dimension $i\in\{1,\dots,{n_x}\}$ by:
\begin{align*}
R_i(t_0,[\underline{x},\overline{x}],[\underline{p},\overline{p}])\subseteq
[&F(t_0,\underline{\xi}^i,\underline{\pi}^i)-\alpha^i(\underline{\xi}^i-\overline{\xi}^i)-\beta^i(\underline{\pi}^i-\overline{\pi}^i),\\
&F(t_0,\overline{\xi}^i,\overline{\pi}^i)+\alpha^i(\underline{\xi}^i-\overline{\xi}^i)+\beta^i(\underline{\pi}^i-\overline{\pi}^i)].
\end{align*}
\end{proposition}

\paragraph{Remarks}
Similarly to the continuous-time mixed-monotonicity in Section~\ref{sub mixed ct}, Proposition~\ref{prop dt mixed mono} encompasses the method for discrete-time monotone systems as a particular case.
In addition, for any discrete-time system with sign-stable Jacobian matrices (i.e.\ for monotone~\cite{hirsch2005monotone} and mixed-monotone systems as in~\cite{coogan2015efficient}), Proposition~\ref{prop dt mixed mono} returns a tight over-approximation of the reachable set.

\section{Toolbox description}
\label{sec toolbox}
The architecture of the toolbox TIRA is sketched in Figure~\ref{fig archi}.
Its philosophy is to provide a library of interval-based reachability methods that can all be accessed through a unique and simple interface function.
On one side of this interface is the user-provided definition of the reachability problem (time range and intervals of initial states and inputs).
On the other side are each of the over-approximation methods described in Section~\ref{sec reachability} and defined in separate functions.
Therefore, this interface function works as a hub that does not only call the over-approximation method requested by the user, but also checks beforehand if the considered system meets all the requirements for the application of this method.

Several over-approximation methods can then easily be tried to solve the same reachability problem by repeating the same call of this interface after changing the parameter defining the method choice.
If the user does not request a specific method, the interface function picks the most suitable method (following the order in Section~\ref{sec reachability} and Algorithm~\ref{algo hub}) based on the optional system information provided by the user (e.g.\ signs or bounds of the Jacobian matrices).

\begin{algorithm2e}[tbh]
  \SetKwFunction{isDefined}{isDefined}
  \KwIn{$\dot x=f(t,x,p)$ or $x^+=F(t,x,p)$, $t_0$, ($t_f$), $[\underline{x},\overline{x}]$, $[\underline{p},\overline{p}]$}
  \eIf($\backslash\backslash$Continuous-time methods){\isDefined$(t_f)$}{
    \lIf($\backslash\backslash$C/GB){Assumption~\ref{assum growth bound}}{Proposition~\ref{prop growth bound}}
    \lElseIf($\backslash\backslash$CTMM){Assumption~\ref{assum ct mixed mono}}{Proposition~\ref{prop ct mixed mono}}
    \lElse($\backslash\backslash$SDMM: sampling and falsification){Proposition~\ref{prop sensi}}
  }($\backslash\backslash$Discrete-time methods){
    \lIf($\backslash\backslash$DTMM){Assumption~\ref{assum dt mixed mono}}{Proposition~\ref{prop dt mixed mono}}
  }
  \KwOut{Over-approximation $[\underline{R},\overline{R}]$ of $R(t_f;t_0,[\underline{x},\overline{x}],[\underline{p},\overline{p}])$}
\caption{Architecture of the hub function $TIRA$.\label{algo hub}}
\end{algorithm2e}

In addition, the main benefit of the chosen architecture for TIRA is its extensibility.
Indeed, while the four methods from Section~\ref{sec reachability} implemented in TIRA cover a wide range of systems, we do not claim that all existing interval-based reachability methods are included in TIRA.
Since the toolbox is written in Matlab and is thus platform independent and does not require an installation, the users can then easily extend this tool library by defining their own over-approximation method in a separate function and adding its call anywhere in the hub function described in Algorithm~\ref{algo hub}.

We end this brief description of the toolbox architecture by a summary of the required and optional user inputs mentioned above and sketched in Figure~\ref{fig archi}.
\begin{itemize}
\item \emph{Required}: system description as in (\ref{eq system ct}) or (\ref{eq system dt}); definition of Problem~\ref{pb ct} ($t_0$, $t_f$, $[\underline x,\overline x]$, $[\underline p,\overline p]$) or~\ref{pb dt} ($t_0$, $[\underline x,\overline x]$, $[\underline p,\overline p]$).
\item \emph{Recommended}: additional system information used by some methods (signs and bounds of the Jacobians and sensitivities, contraction matrix, growth bound function). If none is provided, TIRA calls the sampled-data mixed-monotonicity approach in Section~\ref{sub sensitivity} using the sampling and falsification method to approximate the sensitivity bounds.
\item \emph{Optional}: request for a specific method; modification of the default internal parameters for some solvers; add new over-approximation methods designed by the user.
\end{itemize}
\begin{figure}[thb]
\centering
\includegraphics[width=\columnwidth]{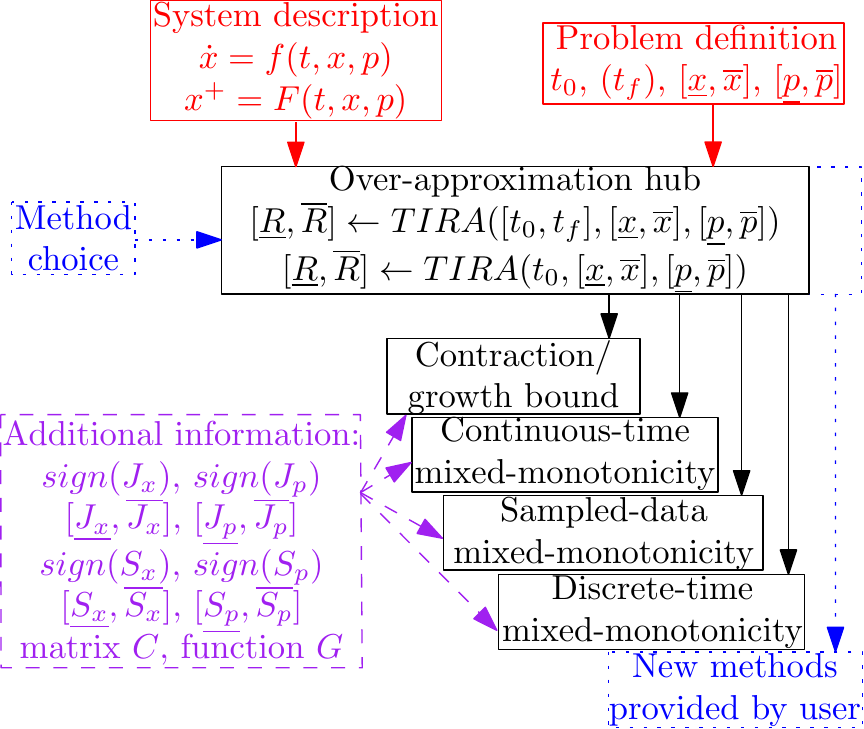}
\caption{TIRA's architecture: black blocks are fully implemented in TIRA; colored blocks are possible user inputs implemented as functions to be filled (required in plain red, recommended in dashed purple, optional in dotted blue).\label{fig archi}}
\end{figure}

\section{Numerical examples}
\label{sec examples}
We consider a $n_x$-link traffic network describing a \emph{diverge} junction (the vehicles in link $1$ divide evenly among the outgoing links $2$ and $3$) followed by downstream links so that traffic on link $2$ flows to link $4$ then to link $6$, \emph{etc.}, and, likewise, traffic flows from link $3$ to $5$ to $7$, \emph{etc.}
Let functions $k:\R^{n_x}\rightarrow\R$ and $l:\R^2\rightarrow\R$ be such that $k(x)=\min(c,vx_1,2w(\bar x-x_2),2w(\bar x-x_3))$ and $l(x_i,x_j)=\min\left(c, vx_i,w(\bar{x}-x_j)/\beta\right)$.
The considered continuous-time model inspired by~\cite{coogan2016benchmark} and written $\dot x=f(x)+p$ as in (\ref{eq system ct additive input}) is then given by:
\begin{align*}
f_1(x)&=-k(x)/T,\\
f_{i}(x)&=(k(x)/2-l(x_i,x_{i+2}))/T,\ &i\in\{2,3\}\\
f_{i}(x)&=(\beta l(x_{i-2},x_i)-l(x_i,x_{i+2}))/T,\ &i\in\{4,\ldots,n\}
\end{align*}
where the term $w(\bar{x}-x_{i+2})/\beta$ is excluded from the minimization in $l$ for $i\in\{n_x-1,n_x\}$.
State $x\in\R^{n_x}$ is the vehicle density on each link, input $p\in\R^{n_x}$ is such that $p_1\in[4/3,2]$ is the constant but uncertain vehicle inflow to link $1$ and $p_i=0$ for $i\geq 2$, and the known parameters of the network $T=30$, $c=40$, $v=0.5$, $\bar x=320$, $w=1/6$ and $\beta=3/4$ are taken from~\cite{coogan2016benchmark}.
Based on these dynamics, we provided to TIRA global bounds for the Jacobian matrices (omitted in this paper due to space limitation).

For the purpose of visualization of the results, we first consider $n_x=3$ and run a function trying all the main over-approximation methods implemented in TIRA with an interval of initial states defined by $\underline{x}=[150;180;100]$ and $\overline{x}=[200;300;220]$.
The methods based on contraction/growth bound, continuous-time mixed-monotonicity and sampled-data mixed-monotonicity (with both interval arithmetic and sampling/falsification submethods to obtain bounds of the sensitivities matrices) are then successfully run with computation times as reported in Table~\ref{tab time}.
The method in Section~\ref{sub mixed dt} is skipped since we do not have a discrete-time system.
The projection onto the $(x_1,x_2)$-plane of the four over-approximations is showed in Figure~\ref{fig traffic OA} alongside an approximation of the actual reachable set by the black cloud of $1000$ sample successor states.

To compare these results with another set representation, we applied the zonotope-based method from CORA~\cite{althoff2015introduction} to the same reachability problem with a similar $3$-link network (taking the smooth approximation $\min(a,b)\approx-\log(e^{-a}+e^{-b})$ since the $\min$ operator cannot be used in CORA's symbolic implementation).
CORA solves the reachability problem by decomposing it into a sequence of intermediate reachability analysis between $t_0=0$ and $t_f=30$s.
At each step, CORA linearizes the nonlinear dynamics and if the considered set is too large, it is iteratively split to keep a low linearization error.
For these reasons and due to our large interval of initial states, CORA was unable to go further than the time instant $18.3$s after $5$ hours of computation\,\footnote{Reusing the solver parameters from CORA's vanDerPol example (\url{https://tumcps.github.io/CORA/}) apart from $timeStep=0.3$ and $maxError=[10;10;10]$.}.
It is plausible that the performance of CORA in this example could be improved with the choice of the internal solver parameters or by avoiding the use of the smoothed version of $\min$\,\footnote{The alternative (not yet attempted) would be to translate the system into a hybrid automaton. For $n_x=3$, this would require $16$ discrete locations and $80$ transitions.}.
TIRA, on the other hand, requires little to no parameter tuning from the user and it does not need the dynamics to be continuously differentiable.

\begin{figure}[htb]
\centering
\includegraphics[width=\columnwidth]{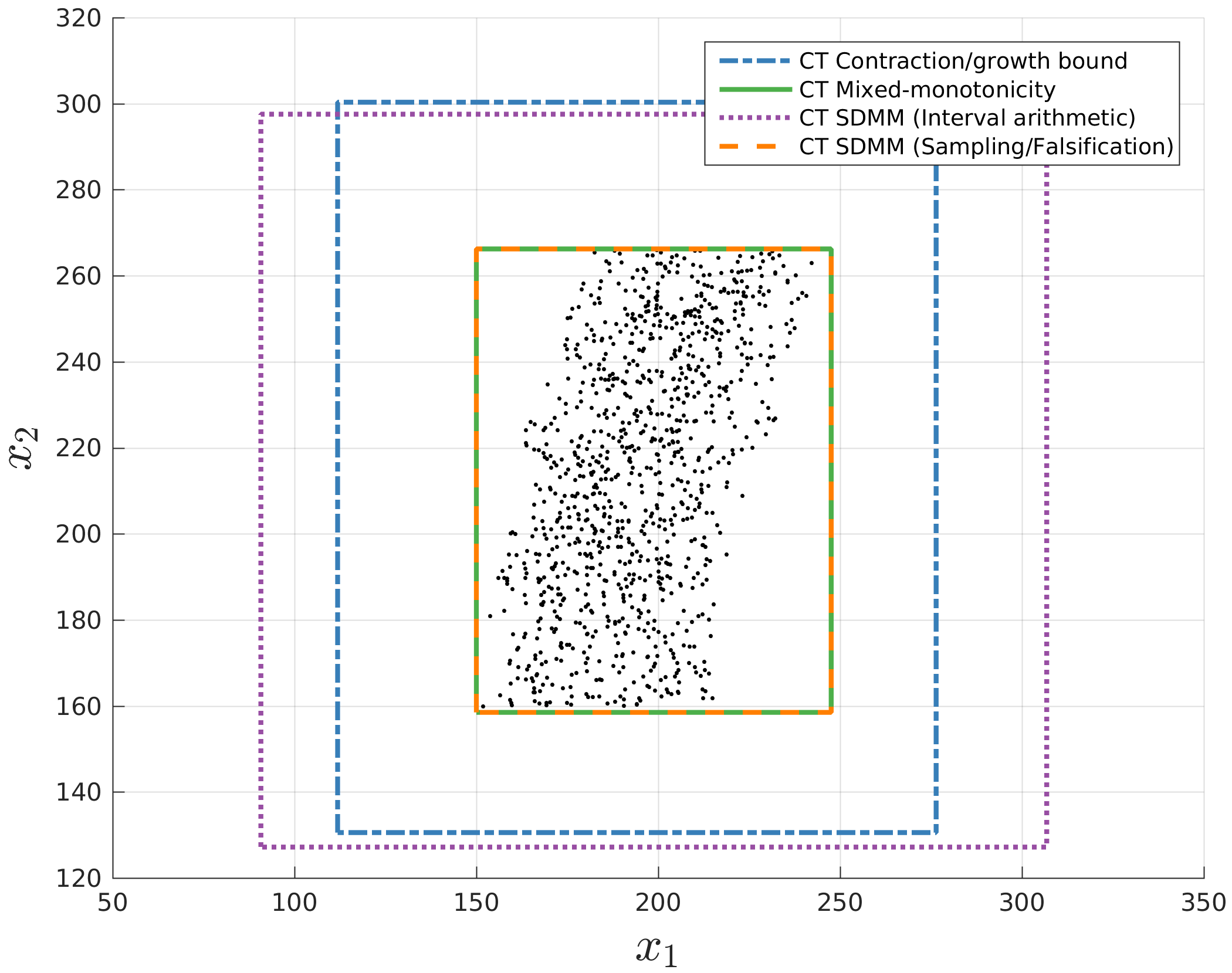}
\caption{Comparison of four over-approximations for the continuous-time model of a $3$-link traffic network representing a \emph{diverge} junction (colored intervals) with an approximation of its reachable set (black cloud of sample successors).\label{fig traffic OA}}
\end{figure}

To evaluate the scalability of the over-approximation methods, we now consider the $n_x$-link network with $n_x=99$ and interval of initial states $[\underline{x},\overline{x}]=[100,200]^{n_x}$.
The sampling and falsification submethod for sampled-data mixed-monotonicity in Section~\ref{sub sensitivity} is discarded from this test since it does not scale to this dimension because the number of samples should grow exponentially with $n_x$ to obtain a decent estimation of the sensitivity bounds.
The computation times for the other three methods are given in Table~\ref{tab time}.
Although the sampled-data mixed-monotonicity approach (with interval arithmetic submethod) appears to have a much worse scalability than the other two, it should be noted that most of its computation time corresponds to the interval arithmetic evaluating the Taylor series of a $n_x\times n_x$ interval matrix exponential ($332$ seconds), while the reachable set over-approximation itself (as in Proposition~\ref{prop sensi}) only takes $5.4$ seconds.

\begin{table}[!t]
\begin{center}
\begin{tabular}{|c||c|c|c|c|c|}
\hline
$n_x$ & C/GB & MM & SDMM (IA) & SDMM (S/F) & CORA\\
\hline
$3$ & $0.13$ & $0.050$ & $0.28$ & $7.0$ & ($>18000$) \\
$99$ & $0.37$ & $4.4$ & $338$ & - & - \\
\hline
\end{tabular}
\medskip
\caption{\label{tab time} Computation times (in seconds) for the over-approximation methods with $n_x=3$ and $n_x=99$, on a laptop with a $1.7$GHz CPU and $4$GB of RAM.}
\end{center}
\end{table}

\section{Conclusions and future work}
\label{sec conclusion}
In this paper, we introduced TIRA, a tool library gathering several methods to over-approximate the reachable set of continuous- and discrete-time systems by a multi-dimensional interval.
Compared to other tools and reachability approaches primarily aimed at the accuracy of over-approximations, TIRA shifts the focus towards the simplicity and scalability of interval methods, some of which providing tight interval over-approximations.
The main feature of TIRA's architecture is to be easily extensible by users who can add their own interval-based reachability methods.

The main directions for future development of TIRA include exploring interval reachability methods for hybrid systems and using existing interval arithmetic tools (see e.g., IBEX~\cite{chabert2009contractor}) to compute Jacobian bounds automatically without requiring user inputs.
Comparing the performances of TIRA to other interval-based tools such as DynIBEX~\cite{dynibex} and VNODE-LP~\cite{nedialkov2006vnode} will also be considered.


\bibliographystyle{ACM-Reference-Format}
\bibliography{2019_HSCC}


\begin{thebibliography}{27}


\ifx \showCODEN    \undefined \def \showCODEN     #1{\unskip}     \fi
\ifx \showDOI      \undefined \def \showDOI       #1{#1}\fi
\ifx \showISBNx    \undefined \def \showISBNx     #1{\unskip}     \fi
\ifx \showISBNxiii \undefined \def \showISBNxiii  #1{\unskip}     \fi
\ifx \showISSN     \undefined \def \showISSN      #1{\unskip}     \fi
\ifx \showLCCN     \undefined \def \showLCCN      #1{\unskip}     \fi
\ifx \shownote     \undefined \def \shownote      #1{#1}          \fi
\ifx \showarticletitle \undefined \def \showarticletitle #1{#1}   \fi
\ifx \showURL      \undefined \def \showURL       {\relax}        \fi
\providecommand\bibfield[2]{#2}
\providecommand\bibinfo[2]{#2}
\providecommand\natexlab[1]{#1}
\providecommand\showeprint[2][]{arXiv:#2}

\bibitem[\protect\citeauthoryear{Althoff}{Althoff}{2015}]%
        {althoff2015introduction}
\bibfield{author}{\bibinfo{person}{Matthias Althoff}.}
  \bibinfo{year}{2015}\natexlab{}.
\newblock \showarticletitle{An Introduction to CORA 2015}. In
  \bibinfo{booktitle}{\emph{ARCH@ CPSWeek}}. \bibinfo{pages}{120--151}.
\newblock


\bibitem[\protect\citeauthoryear{Althoff, Stursberg, and Buss}{Althoff
  et~al\mbox{.}}{2007}]%
        {althoff2007reachability}
\bibfield{author}{\bibinfo{person}{Matthias Althoff}, \bibinfo{person}{Olaf
  Stursberg}, {and} \bibinfo{person}{Martin Buss}.}
  \bibinfo{year}{2007}\natexlab{}.
\newblock \showarticletitle{Reachability analysis of linear systems with
  uncertain parameters and inputs}. In \bibinfo{booktitle}{\emph{46th IEEE
  Conference on Decision and Control}}. IEEE, \bibinfo{pages}{726--732}.
\newblock


\bibitem[\protect\citeauthoryear{Angeli and Sontag}{Angeli and Sontag}{2003}]%
        {angeli2003monotone}
\bibfield{author}{\bibinfo{person}{David Angeli} {and}
  \bibinfo{person}{Eduardo~D. Sontag}.} \bibinfo{year}{2003}\natexlab{}.
\newblock \showarticletitle{Monotone Control Systems}.
\newblock \bibinfo{journal}{\emph{IEEE Trans. Automat. Control}}
  \bibinfo{volume}{48}, \bibinfo{number}{10} (\bibinfo{year}{2003}),
  \bibinfo{pages}{1684--1698}.
\newblock


\bibitem[\protect\citeauthoryear{Arcak and Maidens}{Arcak and Maidens}{2018}]%
        {ArcMai18EALfest}
\bibfield{author}{\bibinfo{person}{M. Arcak} {and} \bibinfo{person}{J.
  Maidens}.} \bibinfo{year}{2018}\natexlab{}.
\newblock \showarticletitle{Simulation-based reachability analysis for
  nonlinear systems using componentwise contraction properties}.
\newblock In \bibinfo{booktitle}{\emph{Principles of Modeling}},
  \bibfield{editor}{\bibinfo{person}{M.~Lohstroh}, \bibinfo{person}{P.~Derler},
  {and} \bibinfo{person}{M.~Sirjani}} (Eds.). \bibinfo{publisher}{Springer},
  \bibinfo{pages}{61--76}.
\newblock


\bibitem[\protect\citeauthoryear{Chabert and Jaulin}{Chabert and
  Jaulin}{2009}]%
        {chabert2009contractor}
\bibfield{author}{\bibinfo{person}{Gilles Chabert} {and} \bibinfo{person}{Luc
  Jaulin}.} \bibinfo{year}{2009}\natexlab{}.
\newblock \showarticletitle{Contractor programming}.
\newblock \bibinfo{journal}{\emph{Artificial Intelligence}}
  \bibinfo{volume}{173} (\bibinfo{year}{2009}), \bibinfo{pages}{1079--1100}.
\newblock


\bibitem[\protect\citeauthoryear{Chen, {\'A}brah{\'a}m, and
  Sankaranarayanan}{Chen et~al\mbox{.}}{2013}]%
        {chen2013flow}
\bibfield{author}{\bibinfo{person}{Xin Chen}, \bibinfo{person}{Erika
  {\'A}brah{\'a}m}, {and} \bibinfo{person}{Sriram Sankaranarayanan}.}
  \bibinfo{year}{2013}\natexlab{}.
\newblock \showarticletitle{Flow*: An analyzer for non-linear hybrid systems}.
  In \bibinfo{booktitle}{\emph{International Conference on Computer Aided
  Verification}}. Springer, \bibinfo{pages}{258--263}.
\newblock


\bibitem[\protect\citeauthoryear{Chu and Huang}{Chu and Huang}{1998}]%
        {chu1998mixed}
\bibfield{author}{\bibinfo{person}{Tianguang Chu} {and} \bibinfo{person}{Lin
  Huang}.} \bibinfo{year}{1998}\natexlab{}.
\newblock \showarticletitle{Mixed monotone decomposition of dynamical systems
  with application}.
\newblock \bibinfo{journal}{\emph{Chinese science bulletin}}
  \bibinfo{volume}{43}, \bibinfo{number}{14} (\bibinfo{year}{1998}),
  \bibinfo{pages}{1171--1175}.
\newblock


\bibitem[\protect\citeauthoryear{Coogan and Arcak}{Coogan and Arcak}{[n. d.]}]%
        {coogan2016benchmark}
\bibfield{author}{\bibinfo{person}{Samuel Coogan} {and} \bibinfo{person}{Murat
  Arcak}.} \bibinfo{year}{[n. d.]}\natexlab{}.
\newblock \showarticletitle{A Benchmark Problem in Transportation Networks}.
\newblock \bibinfo{journal}{\emph{arXiv preprint arXiv:1803.00367}}
  (\bibinfo{year}{[n. d.]}).
\newblock


\bibitem[\protect\citeauthoryear{Coogan and Arcak}{Coogan and Arcak}{2015}]%
        {coogan2015efficient}
\bibfield{author}{\bibinfo{person}{Samuel Coogan} {and} \bibinfo{person}{Murat
  Arcak}.} \bibinfo{year}{2015}\natexlab{}.
\newblock \showarticletitle{Efficient finite abstraction of mixed monotone
  systems}. In \bibinfo{booktitle}{\emph{18th International Conference on
  Hybrid Systems: Computation and Control}}. ACM, \bibinfo{pages}{58--67}.
\newblock


\bibitem[\protect\citeauthoryear{Coogan, Arcak, and Kurzhanskiy}{Coogan
  et~al\mbox{.}}{2016}]%
        {coogan2016mixed}
\bibfield{author}{\bibinfo{person}{Samuel Coogan}, \bibinfo{person}{Murat
  Arcak}, {and} \bibinfo{person}{Alexander~A. Kurzhanskiy}.}
  \bibinfo{year}{2016}\natexlab{}.
\newblock \showarticletitle{Mixed monotonicity of partial first-in-first-out
  traffic flow models}. In \bibinfo{booktitle}{\emph{55th IEEE Conference on
  Decision and Control}}. IEEE, \bibinfo{pages}{7611--7616}.
\newblock


\bibitem[\protect\citeauthoryear{{dit Sandretto} and Chapoutot}{{dit Sandretto}
  and Chapoutot}{2016}]%
        {dynibex}
\bibfield{author}{\bibinfo{person}{Julien~Alexandre {dit Sandretto}} {and}
  \bibinfo{person}{Alexandre Chapoutot}.} \bibinfo{year}{2016}\natexlab{}.
\newblock \showarticletitle{Validated Explicit and Implicit {R}unge--{K}utta
  Methods}.
\newblock \bibinfo{journal}{\emph{Reliable Computing}} \bibinfo{volume}{22},
  \bibinfo{number}{1} (\bibinfo{date}{Jul} \bibinfo{year}{2016}),
  \bibinfo{pages}{79--103}.
\newblock


\bibitem[\protect\citeauthoryear{Dreossi}{Dreossi}{2017}]%
        {dreossi2017sapo}
\bibfield{author}{\bibinfo{person}{Tommaso Dreossi}.}
  \bibinfo{year}{2017}\natexlab{}.
\newblock \showarticletitle{Sapo: reachability computation and parameter
  synthesis of polynomial dynamical systems}. In \bibinfo{booktitle}{\emph{20th
  International Conference on Hybrid Systems: Computation and Control}}. ACM,
  \bibinfo{pages}{29--34}.
\newblock


\bibitem[\protect\citeauthoryear{Frehse, Le~Guernic, Donz{\'e}, Cotton, Ray,
  Lebeltel, Ripado, Girard, Dang, and Maler}{Frehse et~al\mbox{.}}{2011}]%
        {frehse2011spaceex}
\bibfield{author}{\bibinfo{person}{Goran Frehse}, \bibinfo{person}{Colas
  Le~Guernic}, \bibinfo{person}{Alexandre Donz{\'e}}, \bibinfo{person}{Scott
  Cotton}, \bibinfo{person}{Rajarshi Ray}, \bibinfo{person}{Olivier Lebeltel},
  \bibinfo{person}{Rodolfo Ripado}, \bibinfo{person}{Antoine Girard},
  \bibinfo{person}{Thao Dang}, {and} \bibinfo{person}{Oded Maler}.}
  \bibinfo{year}{2011}\natexlab{}.
\newblock \showarticletitle{SpaceEx: Scalable verification of hybrid systems}.
  In \bibinfo{booktitle}{\emph{International Conference on Computer Aided
  Verification}}. Springer, \bibinfo{pages}{379--395}.
\newblock


\bibitem[\protect\citeauthoryear{Hirsch and Smith}{Hirsch and Smith}{2005}]%
        {hirsch2005monotone}
\bibfield{author}{\bibinfo{person}{Morris~W. Hirsch} {and} \bibinfo{person}{Hal
  Smith}.} \bibinfo{year}{2005}\natexlab{}.
\newblock \showarticletitle{Monotone maps: a review}.
\newblock \bibinfo{journal}{\emph{Journal of Difference Equations and
  Applications}} \bibinfo{volume}{11}, \bibinfo{number}{4-5}
  (\bibinfo{year}{2005}), \bibinfo{pages}{379--398}.
\newblock


\bibitem[\protect\citeauthoryear{Jaulin, Kieffer, Didrit, and Walter}{Jaulin
  et~al\mbox{.}}{2001}]%
        {jaulin2001applied}
\bibfield{author}{\bibinfo{person}{Luc Jaulin}, \bibinfo{person}{Michel
  Kieffer}, \bibinfo{person}{Olivier Didrit}, {and} \bibinfo{person}{Eric
  Walter}.} \bibinfo{year}{2001}\natexlab{}.
\newblock \bibinfo{booktitle}{\emph{Applied interval analysis: with examples in
  parameter and state estimation, robust control and robotics}}.
  Vol.~\bibinfo{volume}{1}.
\newblock \bibinfo{publisher}{Springer Science \& Business Media}.
\newblock


\bibitem[\protect\citeauthoryear{Kapela and Zgliczy{\'n}ski}{Kapela and
  Zgliczy{\'n}ski}{2009}]%
        {kapela2009lohner}
\bibfield{author}{\bibinfo{person}{Tomasz Kapela} {and} \bibinfo{person}{Piotr
  Zgliczy{\'n}ski}.} \bibinfo{year}{2009}\natexlab{}.
\newblock \showarticletitle{A Lohner-type algorithm for control systems and
  ordinary differential inclusions}.
\newblock \bibinfo{journal}{\emph{Discrete and Continuous Dynamical Systems.
  Series B}} \bibinfo{volume}{11}, \bibinfo{number}{2} (\bibinfo{year}{2009}),
  \bibinfo{pages}{365--385}.
\newblock


\bibitem[\protect\citeauthoryear{Kurzhanskiy and Varaiya}{Kurzhanskiy and
  Varaiya}{2006}]%
        {kurzhanskiy2006ellipsoidal}
\bibfield{author}{\bibinfo{person}{Alex~A. Kurzhanskiy} {and}
  \bibinfo{person}{Pravin Varaiya}.} \bibinfo{year}{2006}\natexlab{}.
\newblock \showarticletitle{Ellipsoidal toolbox (ET)}. In
  \bibinfo{booktitle}{\emph{45th IEEE Conference on Decision and Control}}.
  IEEE, \bibinfo{pages}{1498--1503}.
\newblock


\bibitem[\protect\citeauthoryear{Mazo~Jr, Davitian, and Tabuada}{Mazo~Jr
  et~al\mbox{.}}{2010}]%
        {pessoa}
\bibfield{author}{\bibinfo{person}{Manuel Mazo~Jr}, \bibinfo{person}{Anna
  Davitian}, {and} \bibinfo{person}{Paulo Tabuada}.}
  \bibinfo{year}{2010}\natexlab{}.
\newblock \showarticletitle{PESSOA: A tool for embedded controller synthesis.}.
  In \bibinfo{booktitle}{\emph{International Conference on Computer Aided
  Verification}}. Springer, \bibinfo{pages}{566--569}.
\newblock


\bibitem[\protect\citeauthoryear{Meyer, Coogan, and Arcak}{Meyer
  et~al\mbox{.}}{2018}]%
        {meyer2018sampled}
\bibfield{author}{\bibinfo{person}{Pierre-Jean Meyer}, \bibinfo{person}{Samuel
  Coogan}, {and} \bibinfo{person}{Murat Arcak}.}
  \bibinfo{year}{2018}\natexlab{}.
\newblock \showarticletitle{Sampled-data reachability analysis using
  sensitivity and mixed-monotonicity}.
\newblock \bibinfo{journal}{\emph{IEEE Control Systems Letters}}
  \bibinfo{volume}{2}, \bibinfo{number}{4} (\bibinfo{year}{2018}),
  \bibinfo{pages}{761--766}.
\newblock


\bibitem[\protect\citeauthoryear{Meyer and Dimarogonas}{Meyer and
  Dimarogonas}{2018}]%
        {meyer2017hierarchical}
\bibfield{author}{\bibinfo{person}{Pierre-Jean Meyer} {and}
  \bibinfo{person}{Dimos~V. Dimarogonas}.} \bibinfo{year}{2018}\natexlab{}.
\newblock \showarticletitle{Hierarchical decomposition of {LTL} synthesis
  problem for nonlinear control systems}.
\newblock \bibinfo{journal}{\emph{arXiv preprint arXiv:1712.06014}}
  (\bibinfo{year}{2018}).
\newblock


\bibitem[\protect\citeauthoryear{Mitchell and Templeton}{Mitchell and
  Templeton}{2005}]%
        {mitchell2005toolbox}
\bibfield{author}{\bibinfo{person}{Ian~M. Mitchell} {and}
  \bibinfo{person}{Jeremy~A. Templeton}.} \bibinfo{year}{2005}\natexlab{}.
\newblock \showarticletitle{A toolbox of Hamilton-Jacobi solvers for analysis
  of nondeterministic continuous and hybrid systems}. In
  \bibinfo{booktitle}{\emph{International Workshop on Hybrid Systems:
  Computation and Control}}. Springer, \bibinfo{pages}{480--494}.
\newblock


\bibitem[\protect\citeauthoryear{Moor and Raisch}{Moor and Raisch}{2002}]%
        {moor2002abstraction}
\bibfield{author}{\bibinfo{person}{Thomas Moor} {and} \bibinfo{person}{J{\"o}rg
  Raisch}.} \bibinfo{year}{2002}\natexlab{}.
\newblock \showarticletitle{Abstraction based supervisory controller synthesis
  for high order monotone continuous systems}.
\newblock In \bibinfo{booktitle}{\emph{Modelling, Analysis, and Design of
  Hybrid Systems}}. \bibinfo{publisher}{Springer}, \bibinfo{pages}{247--265}.
\newblock


\bibitem[\protect\citeauthoryear{Nedialkov}{Nedialkov}{2006}]%
        {nedialkov2006vnode}
\bibfield{author}{\bibinfo{person}{Ned Nedialkov}.}
  \bibinfo{year}{2006}\natexlab{}.
\newblock \bibinfo{booktitle}{\emph{VNODE-LP}}.
\newblock \bibinfo{type}{{T}echnical {R}eport} CAS-06-06-NN.
  \bibinfo{institution}{Dept. of Computing and Software, McMaster Univ.
  Hamilton, ON, Canada}.
\newblock


\bibitem[\protect\citeauthoryear{Reissig, Weber, and Rungger}{Reissig
  et~al\mbox{.}}{2016}]%
        {reissig2016}
\bibfield{author}{\bibinfo{person}{Gunther Reissig}, \bibinfo{person}{Alexander
  Weber}, {and} \bibinfo{person}{Matthias Rungger}.}
  \bibinfo{year}{2016}\natexlab{}.
\newblock \showarticletitle{Feedback refinement relations for the synthesis of
  symbolic controllers}.
\newblock \bibinfo{journal}{\emph{IEEE Trans. Automat. Control}}
  \bibinfo{volume}{62}, \bibinfo{number}{4} (\bibinfo{year}{2016}),
  \bibinfo{pages}{1781--1796}.
\newblock


\bibitem[\protect\citeauthoryear{Rungger and Zamani}{Rungger and
  Zamani}{2016}]%
        {scots}
\bibfield{author}{\bibinfo{person}{Matthias Rungger} {and}
  \bibinfo{person}{Majid Zamani}.} \bibinfo{year}{2016}\natexlab{}.
\newblock \showarticletitle{SCOTS: A tool for the synthesis of symbolic
  controllers}. In \bibinfo{booktitle}{\emph{Proceedings of the 19th
  International Conference on Hybrid Systems: Computation and Control}}. ACM,
  \bibinfo{pages}{99--104}.
\newblock


\bibitem[\protect\citeauthoryear{Xue, Fr{\"a}nzle, and Mosaad}{Xue
  et~al\mbox{.}}{2017}]%
        {xue2017cdc}
\bibfield{author}{\bibinfo{person}{Bai Xue}, \bibinfo{person}{Martin
  Fr{\"a}nzle}, {and} \bibinfo{person}{Peter~Nazier Mosaad}.}
  \bibinfo{year}{2017}\natexlab{}.
\newblock \showarticletitle{Just Scratching the Surface: Partial Exploration of
  Initial Values in Reach-Set Computation}. In \bibinfo{booktitle}{\emph{56th
  IEEE Conference on Decision and Control}}. \bibinfo{pages}{1769--1775}.
\newblock


\bibitem[\protect\citeauthoryear{Yang, Mickelin, and Ozay}{Yang
  et~al\mbox{.}}{2018}]%
        {yang2018sufficient}
\bibfield{author}{\bibinfo{person}{Liren Yang}, \bibinfo{person}{Oscar
  Mickelin}, {and} \bibinfo{person}{Necmiye Ozay}.}
  \bibinfo{year}{2018}\natexlab{}.
\newblock \showarticletitle{On sufficient conditions for mixed monotonicity}.
\newblock \bibinfo{journal}{\emph{arXiv preprint arXiv:1803.04528}}
  (\bibinfo{year}{2018}).
\newblock


\end{thebibliography}

\ifproof
\appendix

\newpage
\section{Continuous-time monotonicity}
\label{appendix mono}
This section presents an over-approximation method which is only applicable to systems satisfying a monotonicity property defined below.
While this method is also available in TIRA, it is not presented in Section~\ref{sec reachability} of this paper because the continuous-time mixed-monotonicity approach in Section~\ref{sub mixed ct} encompasses it as a particular case.
Further comments on the comparison of these two methods are provided at the end of this section.

\paragraph{Requirements and limitations}
The monotonicity property for continuous-time systems with inputs (\ref{eq system ct}) is defined in~\cite{angeli2003monotone} and used for reachability analysis in~\cite{moor2002abstraction}.
A system (\ref{eq system ct}) is monotone if its Jacobian matrices $J_x(t,x,p)$ and $J_p(t,x,p)$ are sign-stable (apart from the diagonal of $J_x$) over the considered ranges of time, state and input and the sign structure satisfies the following assumption.
\begin{assum}
\label{assum mono}
Given an invariant state space $X\subseteq\R^{n_x}$, there exist $\varepsilon=[\varepsilon_1;\dots;\varepsilon_{n_x}]\in\{0,1\}^{n_x}$ and $\delta=[\delta_1;\dots;\delta_{n_p}]\in\{0,1\}^{n_p}$ such that for all $t\in[t_0,t_f]$, $x\in X$, $p\in[\underline{p},\overline{p}]$, $i,j\in\{1,\dots,n_x\}$, $j\neq i$ and $k\in\{1,\dots,n_p\}$ we have:
\begin{equation*}
(-1)^{\varepsilon_i+\varepsilon_j}\frac{\partial f_i(t,x,p)}{x_j}\geq0,\qquad
(-1)^{\varepsilon_i+\delta_k}\frac{\partial f_i(t,x,p)}{p_k}\geq0.
\end{equation*}
\end{assum}
Note that the user does not need to know in advance whether their system is monotone since TIRA automatically checks this sign structure by translating Assumption~\ref{assum mono} into a system of boolean equations and solving it in the 2-element Galois Field GF(2).

\paragraph{Method description}
An over-approximation of the reachable set is computed by evaluating the successor states of (\ref{eq system ct}) for only two pairs $(x,p)\in[\underline{x},\overline{x}]\times[\underline{p},\overline{p}]$ picked based on the boolean vectors $\varepsilon=[\varepsilon_1;\dots;\varepsilon_{n_x}]$ and $\delta=[\delta_1;\dots;\delta_{n_p}]$ satisfying Assumption~\ref{assum mono}.
\begin{proposition}
\label{prop mono}
Under Assumption~\ref{assum mono}, an over-approximation of the reachable set of (\ref{eq system ct}) in Problem~\ref{pb ct} is given by (using componentwise multiplications with $\varepsilon$ and $\delta$):
\begin{multline*}
R(t_f;t_0,[\underline{x},\overline{x}],[\underline{p},\overline{p}])\subseteq\\
[\Phi(t_f;t_0,\underline{x}(1_{n_x}-\varepsilon)+\overline{x}\varepsilon,\underline{p}(1_{n_p}-\delta)+\overline{p}\delta),\\
\Phi(t_f;t_0,\underline{x}\varepsilon+\overline{x}(1_{n_x}-\varepsilon),\underline{p}\delta+\overline{p}(1_{n_p}-\delta))].
\end{multline*}
\end{proposition}

\paragraph{Remarks}
While Assumption~\ref{assum mono} is quite restrictive, whenever it is satisfied the resulting interval in Proposition~\ref{prop mono} is guaranteed to be a tight over-approximation of the reachable set.
As mentioned in Proposition~\ref{prop ct mono} and proved below in Appendix~\ref{appendix sub mono}, applying the continuous-time mixed-monotonicity approach in Proposition~\ref{prop ct mixed mono} to a monotone system satisfying Assumption~\ref{assum mono} will result in the same tight interval over-approximation as in Proposition~\ref{prop mono}.
The main differences between these two results is that the monotonicity-specific result in Proposition~\ref{prop mono} has a constant complexity (we always only evaluate $\Phi$ for two state-input pairs in $[\underline{x},\overline{x}]\times[\underline{p},\overline{p}]$), while the complexity of the more general result in Proposition~\ref{prop ct mixed mono} is linear in the state dimension $n_x$ ($2n_x$ evaluations of $\Phi$ are required).
On the other hand, Proposition~\ref{prop ct mixed mono} does not need to know whether Assumption~\ref{assum mono} is satisfied to obtain this result, while Proposition~\ref{prop mono} first requires checking Assumption~\ref{assum mono} through the provided function in TIRA which can be time consuming for large systems.

\section{Proofs of Section~\ref{sub mixed ct}}
\label{appendix proofs}
In this section, $\R_+$ and $\R_-$ are the sets of non-negative and non-positive real numbers, respectively.

\subsection{Proposition~\ref{prop ct mixed mono}}

\begin{proof}[Proof of Proposition~\ref{prop ct mixed mono}]
From the definitions of functions $g$ and $h$ in (\ref{eq decomposition function})-(\ref{eq duplicated system}), we have for all $i,j\in\{1,\dots,n_x\}$, $j\neq i$ and $k\in\{1,\dots,n_p\}$:
\begin{align*}
\frac{\partial h_i(t,x,p,\hat x,\hat p)}{\partial x_j}=\frac{\partial f_i(t,\xi^i,\pi^i)}{\partial x_j}+\alpha^i_j\geq0\\
\frac{\partial h_i(t,x,p,\hat x,\hat p)}{\partial \hat x_j}=\frac{\partial f_i(t,\xi^i,\pi^i)}{\partial \hat x_j}-\alpha^i_j\leq0\\
\frac{\partial h_i(t,x,p,\hat x,\hat p)}{\partial \hat x_i}=\frac{\partial f_i(t,\xi^i,\pi^i)}{\partial \hat x_i}-\alpha^i_i=0
\end{align*}
Similarly, we obtain $\frac{\partial h_{n_x+i}}{\partial x_i}=0$, $\frac{\partial h_{n_x+i}}{\partial x_j}\leq0$, $\frac{\partial h_{n_x+i}}{\partial \hat x_j}\geq0$, $\frac{\partial h_i}{\partial p_k}\geq0$, $\frac{\partial h_i}{\partial \hat p_k}\leq0$, $\frac{\partial h_{n_x+i}}{\partial p_k}\leq0$ and $\frac{\partial h_{n_x+i}}{\partial \hat p_k}\geq0$.
This implies that system (\ref{eq duplicated system}) is monotone with respect to the orthants $\R^{n_x}_+\times\R^{n_x}_-$ and $\R^{n_p}_+\times\R^{n_p}_-$.
Then from~\cite{angeli2003monotone}, for all $x\in[\underline{x},\overline{x}]$ and $\mathbf{p}:[t_0,t_f]\rightarrow[\underline{p},\overline{p}]$, we have
$$
\Phi^h(t_f;t_0,\underline x,\underline p,\overline{x},\overline{p})\preceq_x
\Phi^h(t_f;t_0,x,\mathbf{p},x,\mathbf{p})\preceq_x
\Phi^h(t_f;t_0,\overline x,\overline p,\underline{x},\underline{p})
$$
where $\preceq_x$ is the partial order defined by the orthant $\R^{n_x}_+\times\R^{n_x}_-$, (i.e.\ for all $x,\hat x,y,\hat y\in\R^{n_x}$, 
$\begin{pmatrix}x \\ \hat x\end{pmatrix}
\preceq_x
\begin{pmatrix}y \\ \hat y\end{pmatrix}
\Leftrightarrow
\begin{cases}
x\leq y,\\
\hat x\geq \hat y,
\end{cases}$
where $\leq$ and $\geq$ are the componentwise inequalities on $\R^{n_x}$).
From (\ref{eq decomposition function}), $f$ is embedded in the diagonal of $g$ (i.e.\ $g(t,x,p,x,p)=f(t,x,p)$), which implies that $\Phi^h(t_f;t_0,x,\mathbf{p},x,\mathbf{p})=\begin{pmatrix}\Phi(t_f;t_0,x,\mathbf{p})\\\Phi(t_f;t_0,x,\mathbf{p})\end{pmatrix}$.
Finally, the symmetry of system (\ref{eq duplicated system}) implies that $\Phi^h_{n_x+1\dots 2n_x}(t_f;t_0,\overline x,\overline p,\underline{x},\underline{p})=\Phi^h_{1\dots n_x}(t_f;t_0,\underline x,\underline p,\overline{x},\overline{p})$, which results in the proposition statement.
\end{proof}

\subsection{Proposition~\ref{prop ct mono}}
\label{appendix sub mono}
\begin{proof}[Proof of Proposition~\ref{prop ct mono}]
%
%
We start from a system (\ref{eq system ct}) satisfying the monotonicity condition in Assumption~\ref{assum mono}.
Without loss of generality, we assume that the states in $x\in\R^{n_x}$ are ordered as $x=[x^+;x^-]$ with $x^+\in\R^{n_x^+}$, $x^-\in\R^{n_x^-}$, $n_x^++n_x^-=n_x$ and such that $\varepsilon=[0_{n_x^+};1_{n_x^-}]$.
We use similar notations $p^+\in\R^{n_p^+}$, $p^-\in\R^{n_p^-}$ and $\delta=[0_{n_p^+};1_{n_p^-}]$ for the input vector $p\in\R^{n_p}$.
We similarly introduce $f^+$, $f^-$, $\Phi^+$, $\Phi^-$ for the decomposition of the vector field $f$ and trajectory function $\Phi$ respectively, into their $n_x^+$ first and $n_x^-$ last components.

If we now apply the result in Proposition~\ref{prop ct mixed mono} to this monotone system, then for all $i\in\{1,\dots,n_x\}$ we have $\alpha^i=0_{n_x}$, $\beta^i=0_{n_p}$ and 
\begin{equation*}
(\xi^i,\pi^i)=
\begin{cases}
(x(1_{n_x}-\varepsilon)+\hat x\varepsilon,p(1_{n_p}-\delta)+\hat p\delta)\text{ if }\varepsilon_i=0,\\
(x\varepsilon+\hat x(1_{n_x}-\varepsilon),p\delta+\hat p(1_{n_p}-\delta))\text{ if }\varepsilon_i=1,
\end{cases}\\
\end{equation*}
using componentwise multiplications.
As a result, system (\ref{eq duplicated system}) becomes:
\begin{equation}
\label{eq duplicated mono}
\begin{pmatrix}\dot x^+\\\dot x^-\\\dot{\hat x}^+\\\dot{\hat x}^-\end{pmatrix}
=h(t,x,p,\hat x,\hat p)=
\begin{pmatrix}
f^+(t,[x^+;\hat x^-],[p^+;\hat p^-])\\
f^-(t,[\hat x^+;x^-],[\hat p^+;p^-])\\
f^+(t,[\hat x^+;x^-],[\hat p^+;p^-])\\
f^-(t,[x^+;\hat x^-],[p^+;\hat p^-])
\end{pmatrix}.
\end{equation}
Since (\ref{eq duplicated mono}) actually contains two decoupled copies of system (\ref{eq system ct}):
$$
\begin{pmatrix}\dot x^+\\\dot{\hat x}^-\end{pmatrix}
=f(t,[x^+;\hat x^-],[p^+;\hat p^-]),\quad
\begin{pmatrix}\dot{\hat x}^+\\\dot x^-\end{pmatrix}
=f(t,[\hat x^+;x^-],[\hat p^+;p^-]),
$$
it implies that any successor of (\ref{eq duplicated mono}) can be expressed as two successors of (\ref{eq system ct}).
In particular, for the quadruple of initial states and inputs $(\underline{x},\underline{p},\overline{x},\overline{p})$ used in Proposition~\ref{prop ct mixed mono}, we have:
$$
\Phi^h(t_f;t_0,\underline{x},\underline{p},\overline{x},\overline{p})=
\begin{pmatrix}
\Phi^+(t_f;t_0,[\underline{x}^+;\overline{x}^-],[\underline{p}^+;\overline{p}^-])\\
\Phi^-(t_f;t_0,[\overline{x}^+;\underline{x}^-],[\overline{p}^+;\underline{p}^-])\\
\Phi^+(t_f;t_0,[\overline{x}^+;\underline{x}^-],[\overline{p}^+;\underline{p}^-])\\
\Phi^-(t_f;t_0,[\underline{x}^+;\overline{x}^-],[\underline{p}^+;\overline{p}^-])
\end{pmatrix}.
$$
Since $\begin{pmatrix}\underline{x}^+\\\overline{x}^-\end{pmatrix},\begin{pmatrix}\overline{x}^+\\\underline{x}^-\end{pmatrix}\in[\underline{x},\overline{x}]$ and $\begin{pmatrix}\underline{p}^+\\\overline{p}^-\end{pmatrix},\begin{pmatrix}\overline{p}^+\\\underline{p}^-\end{pmatrix}\in[\underline{p},\overline{p}]$, we know that $\Phi(t_f;t_0,[\underline{x}^+;\overline{x}^-],[\underline{p}^+;\overline{p}^-])$ and $\Phi(t_f;t_0,[\overline{x}^+;\underline{x}^-],[\overline{p}^+;\underline{p}^-])$ belong to the actual reachable set $R(t_f;t_0,[\underline{x},\overline{x}],[\underline{p},\overline{p}])$ of (\ref{eq system ct}).
As a result, the interval defined from the $2n_x$ components of $\Phi^h(t_f;t_0,\underline{x},\underline{p},\overline{x},\overline{p})$ in Proposition~\ref{prop ct mixed mono} is necessarily a tight interval over-approximation of the reachable set.

Since a tight interval over-approximation of a set is uniquely defined and the reachability method defined for monotone systems in Proposition~\ref{prop mono} is also known to provide a tight interval over-approximation of the reachable set, we can conclude that both methods provide the same results.
\end{proof}
\fi

\end{document}